\documentclass[aps,pra,twocolumn,10pt,showpacs,
superscriptaddress,nofootinbib]{revtex4-1}
\usepackage{amsmath}
\usepackage{amsthm}
\usepackage{graphicx}
\usepackage{times}
\usepackage{amssymb}
\usepackage{hyperref}
\hypersetup{
colorlinks=true,
linkcolor=red,
citecolor=red,}
\usepackage{textcomp}
\usepackage{xcolor}
\usepackage{enumerate}
\usepackage[final]{pdfpages}
\usepackage{multirow}
\usepackage{tabularx}
\usepackage{dcolumn}
\usepackage[mathscr]{euscript}
\usepackage{mathtools}
\usepackage{needspace}

\newtheorem*{lemma}{Lemma}

\begin{document}
\title{Entanglement and nonclassicality: A mutual impression}
 \author{H. Gholipour}
  \email{Electronic address: hamedgholipoor1987@gmail.com}
 \affiliation{School of Physics, Iran University of Science and Technology, Narmak, Tehran 16846-13114, Iran}
 \author{F. Shahandeh}
 \email{Electronic address: f.shahandeh@uq.edu.au}
\affiliation{Centre for Quantum Computation and Communication Technology, School of Mathematics and Physics, University of Queensland, St Lucia, Queensland 4072, Australia}


\begin{abstract}
We find a sufficient condition to imprint the single-mode bosonic phase-space nonclassicality onto a bipartite state as modal entanglement and vice versa using an arbitrary beam splitter.
Surprisingly, the entanglement produced or detected in this way depends only on the nonclassicality of the marginal input or output states, regardless of their purity and separability.
In this way, our result provides a sufficient condition for generating entangled states of arbitrary high temperature and arbitrary large number of particles.
We also study the evolution of the entanglement within a lossy Mach-Zehnder interferometer and show that unless both modes are totally lost, the entanglement does not diminish.
\end{abstract}


\maketitle

\section{Introduction}
Quantum protocols outperform their classical counterparts by taking advantage of quantum resources, the most well-known of which is quantum entanglement~\cite{Einstein,Wooters,Plenio} used for quantum key distribution~\cite{Schrodinger}, quantum dense coding~\cite{Nielsen}, quantum teleportation~\cite{HorodeckiRMP}, and gravitational wave detection~\cite{LIGO}.
Consequently, it is of great importance to have easy approaches for entanglement generation and detection.

Arguably, the simplest way of generating entanglement is to use a beam splitter (BS).
Whether the input beams are two rays of photons~\cite{Kim,Wolf,Asboth,Jiang} or two rays of atoms~\cite{Pezze}, it is of fundamental interest to ask when does a BS generate entangled outputs.
It is also well-known that a BS transforms separable Fock states into modal entangled binomial states, while it leaves separable coherent states unentangled~\cite{Leonhardt}.
Such relations were first considered by Kim \textit{et al.}~\cite{Kim} and it was proven by Wang~\cite{Wang} that the ability of a BS to produce entangled outputs is closely connected to the nonclassical properties of the input states:
the entanglement in the output of a linear optical network implies the nonclassicality in the phase-space representation of the inputs.
This, indeed, revealed a necessary relation between the two notions of entanglement detection and nonclassicality, which was used to quantify the input nonclassicality via measuring the entanglement of the output~\cite{Asboth}.
In this direction, a nice study relating distillable entanglement detection to photon statistics nonclassicality is given in Ref.~\cite{Solomon}.
Further investigations led to a complete understanding of the entanglement-nonclassicality correspondence from the perspective of Gaussian states and distillable entanglement~\cite{Wolf,Brunelli}. 
Moreover, it has recently been shown that if the Glauber-Sudarshan $P$-functions of pure product inputs to a $N$-port connected linear optical network are nonclassical, then the output will almost always exhibit entanglement~\cite{Jiang}.
Nonetheless, a peculiar situation exists in which the entanglement--nonclassicality correspondence breaks down:
whenever two squeezed states with parallel squeezing axes interact on a balanced beam splitter, they will be transformed into separable product states although they are considered to be highly nonclassical.

To reconsider the mutual relation between entanglement generation (detection) and nonclassicality, we need an appropriate criterion for entanglement detection.
There are several criteria for verification of entanglement (see, e.g., Ref.~\cite{HorodeckiRMP} and references therein), probably the most famous one is the partial transposition criterion~\cite{Peres,Horodecki}. 
There is another extensively studied approach, the so-called entanglement witnessing, which provides a necessary and sufficient condition~\cite{Horodecki,Terhal,HorodeckiWIT,Lewen}. 
The advantage of entanglement witnesses (EWs) is that they allow us to detect entanglement without full information about the quantum state by measuring statistics of a finite number of quantum observables.
Different types of Bell inequalities, for instance, are EWs to test nonlocality~\cite{Terhal}.
Such observables are relatively easy to construct in low dimensions, while generally very hard in high dimensional Hilbert spaces, e.g. for continuous variable (CV) systems.
Despite that there are efficient methods for optimization of EWs in bipartite and even multipartite scenarios~\cite{Sperling1,Sperling2,Shahandeh}, there exist a few methods for EW construction; see, e.g., Ref.~\cite{Sixia}.
Here, however, we use an EW construction to be applied to our particular problem.

In this contribution, we derive a {\it sufficient} condition for the phase-space nonclassicality of at least one of the input states to a BS to guarantee the entanglement of the output modes, regardless of the other input state.
Remarkably, our criterion immediately removes any assumptions about the purity or the separability of the inputs.
Such a simple entanglement generation criterion allows one to easily design schemes for technically difficult tasks such as high particle number entanglement generation.
For example, we prove that a single-photon state of sufficient purity generates entanglement from any state, e.g. an arbitrarily high temperature thermal state, upon incidence on any BS.
We also re-examine entanglement detection stage and show that entanglement can be imprinted onto a single-mode in the form of phase-space negativities.
In other words, detecting negativities of marginal phase-space functions, e.g. using homodyne detection, at the output of a BS provides a simple criterion for local entanglement verification.

The paper is organized as follows.
We start with introducing our EW construction approach in Sec.~\ref{SecII} together with a detailed account of its properties.
In Sec.~\ref{SecIII} we derive our sufficient criterion for transforming nonclassicality to entanglement and vice versa using an arbitrary BS.
The extensive discussion of the results and their implications is given in Sec.~\ref{SecIV}.
To study the effect of loss on the generation and detection of entanglement, we study the explicit example of a Mach-Zehnder interferometer (MZI) with lossy arms in Sec.~\ref{SecV}.
A summary and conclusions are provided in Sec.~\ref{SecVI}.


\section{EW construction}~\label{SecII}

Let us start with introducing our EW construction procedure.
Consider a $n$-dimensional Hilbert space $\mathcal{H}$ 
and the linear vector space $\mathfrak{A}(\mathcal{H})$ of linear operators acting on $\mathcal{H}$, endowed with Hilbert-Schmidt inner product $(\hat{A},\hat{B}):=\text{Tr}\hat{A}^{\dag}\hat{B}$.
Denote the subset of $\mathfrak{A}(\mathcal{H})$ consisting of linear operators with finite norm as $\mathfrak{L}(\mathcal{H})$ and
the space of quantum states $\mathfrak{S}(\mathcal{H})$ as the subset of $\mathfrak{L}(\mathcal{H})$ for which all members are positive and normalized.
Thus, we have the hierarchy of spaces as $\mathfrak{S}(\mathcal{H})\subset\mathfrak{L}(\mathcal{H})\subset\mathfrak{A}(\mathcal{H})$.

Any linear functional from $\mathfrak{L}(\mathcal{H})$ to $\mathbb{C}$ 
is given by ${\rm Tr}\hat{F}\hat{A}$ for some $\hat{F}\in\mathfrak{L}(\mathcal{H})$ for all $\hat{A}\in\mathfrak{L}(\mathcal{H})$ via the celebrated Riesz theorem.
Suppose that the set $\mathfrak{X}=\{\hat{X}_{ij}\}\subset\mathfrak{A}(\mathcal{H})$ ($i,j=1,2,\dots,n$) is any basis set for $\mathfrak{L}(\mathcal{H})$ such that $\hat{X}_{ij}^{\dag}:=\hat{X}_{ji}$.
We call the set of operators $\mathfrak{X}'=\{\hat{X}_{ij}^{\prime}\}$ dual to $\mathfrak{X}$ if
\begin{equation}
\label{orthXdX}
 (\hat{X}_{ij}^{\prime},\hat{X}_{kl})=\text{Tr}\hat{X}_{ij}^{\prime\dag}\hat{X}_{kl}=\delta_{ik}\delta_{jl}.
\end{equation}
This simple construction allows us to expand any quantum state in either of the bases as
\begin{equation}
\label{rhoExp}
\begin{split}
 & \hat{\varrho}=\sum_{ij=1}^{n}\varrho_{ij}^\prime\hat{X}_{ij} \quad\text{with}\quad \varrho_{ij}^\prime=(\hat{X}_{ij}^{\prime},\hat{\varrho})=\text{Tr}\hat{X}_{ij}^{\prime\dag}\hat{\varrho},\\
 & \hat{\varrho}=\sum_{ij=1}^{n}\varrho_{ij}\hat{X}_{ij}^{\prime} \quad\text{with}\quad \varrho_{ij}=(\hat{X}_{ij},\hat{\varrho})=\text{Tr}\hat{X}_{ij}^{\dag}\hat{\varrho}.
\end{split}
\end{equation}
For two arbitrary quantum states $\hat{\varrho}_{1,2}{\in}\mathfrak{S}(\mathcal{H})$,
the Cauchy-Schwarz inequality reads as
\begin{equation}
0\leq\text{Tr}\hat{\varrho}_1\hat{\varrho}_2 \leq \sqrt{\text{Tr}\hat{\varrho}_1^2}\sqrt{\text{Tr}\hat{\varrho}_2^2}\leq 1,
\end{equation}
where we have used the facts that density operators are positive and $\text{Tr}\hat{\varrho}^2\leq 1$ for any quantum state.
Now, expanding $\hat{\varrho}_1$ in $\{\hat{X}_{ij}^{\prime}\}$ basis and $\hat{\varrho}_2$ in $\{\hat{X}_{ij}\}$ basis  using Eqs.~\eqref{orthXdX} and~\eqref{rhoExp} gives
\begin{equation}
\label{SimpCS}
0\leq\text{Tr}\hat{\varrho}_1\hat{\varrho}_2=\sum_{ij=1}^{n} \varrho_{1;ij}^{\ast}\varrho_{2;ij}^{\prime}\leq1,
\end{equation}
where $\varrho_{1;ij}^{\ast}:=(\hat{X}_{ij},\hat{\varrho}_1)^{\ast}=\text{Tr}\hat{X}_{ij}\hat{\varrho}_1$ and $\varrho_{2;ij}^{\prime}:=(\hat{X}^\prime_{ij},\hat{\varrho}_2)=\text{Tr}\hat{X}^{\prime\dag}_{ij}\hat{\varrho}_2$.
In addition, suppose that $0<\mathsf{G}\leq \tilde g$ is a bounded positive definite superoperator defined as
\begin{equation}
\mathsf{G}(\cdot)=\sum_{ij=1}^n g_{ij} \hat{X}_{ij}{\rm Tr}\hat{X}^{\prime\dag}_{ij}(\cdot),
\end{equation}
such that $\mathsf{G}(\hat{X}_{kl})=g_{kl}\hat{X}_{kl}$.
The matrix elements of $\mathsf{G}$ are thus given by $\mathsf{G}_{ij;kl}=\text{Tr}\hat{X}^{\prime\dag}_{ij}\mathsf{G}(\hat{X}_{kl})=g_{ij}\delta_{ik}\delta_{jl}$.
It is clear that the modulus of $g_{ij}$ is bounded by $\tilde g$.
Multiplying each term of the sum in Eq.~\eqref{SimpCS} with $g_{ij}$ preserves the inequality with the upper bound $\tilde g$, and thus,
 \begin{equation}
 \label{CSinequal}
  0\leq\sum_{ij=1}^{n} g_{ij}\varrho_{1;ij}^{\ast}\varrho_{2;ij}^{\prime} \leq \tilde g.
 \end{equation}
Consequently, we state the following EW construction procedure from arbitrary set of local operators. 
\begin{enumerate}[(i)]
\item Choose an arbitrary basis set $\mathfrak{X}=\{\hat{X}_{ij}\}$ ($i,j=1,2,\dots,n$) for $\mathfrak{L}(\mathcal{H})$ corresponding to the set of dual operators $\mathfrak{X}'=\{\hat{X}_{ij}^{\prime}\}$.
\item Choose a bounded positive definite superoperator (pseudo-metric) $0<\mathsf{G}\leq \tilde g$ with matrix elements $\text{Tr}\hat{X}^{\prime\dag}_{ij}\mathsf{G}(\hat{X}_{kl})=g_{ij}\delta_{ik}\delta_{jl}$.
\item Define the Hermitian witness operator to be
 \begin{equation}
 \label{Test}
  \hat{W}:=\sum_{ij=1}^{n} g_{ij}\hat{X}_{ij} \otimes \hat{X}_{ij}^{\prime\dag}.
 \end{equation} 
\end{enumerate}
One can simply verify that a bipartite quantum system living in the $n^4$-dimensional state space $\hat{\varrho}\in\mathfrak{S}(\mathcal{H}^{\otimes 2})$ is entangled if
 \begin{equation}
 \label{Wineq}
  (\hat{\varrho},\hat{W})=\text{Tr}\hat{\varrho}\hat{W}\notin [0,\tilde g].
 \end{equation}
This is because, according to the above discussion, for any separable state of the form $\hat{\sigma}=\sum_{k}p_k\hat{\sigma}_{1;k}\otimes\hat{\sigma}_{2;k}$ with $\sum_{k}p_k{=}1$ one has
$0\leq\text{Tr}\hat{\sigma}\hat{W} \leq \tilde g$.

Let us investigate the properties of the above class of witnesses under local maps.
For this purpose, suppose that the map $\hat{I}_i\otimes\tilde{\Lambda}_j$ ($i,j=1,2$ and $i\neq j$), where $\hat{I}_i$ is the identity of $\mathfrak{A}(\mathcal{H}_i)$ and $\Lambda_j:\mathfrak{A}(\mathcal{H}_j)\rightarrow\mathfrak{A}(\mathcal{H}_j)$, has been applied to the set of separable states so that $\hat{\tau}=\text{Tr}\hat{I}_i\otimes\tilde{\Lambda}_j(\hat{\sigma})$ is an unnormalized separable operator for any separable state $\hat{\sigma}$.
Thus, $0 \leq \text{Tr}\hat{\tau}\hat{W}$ if and only if the map $\tilde{\Lambda}_j$ is a positive map sending quantum states to (unnormalized) positive operators.
Now, corresponding to any map $\tilde{\Lambda}$ the adjoint map $\Lambda$ is defined such that $(\hat{B},\tilde{\Lambda}(\hat{A}))=(\Lambda(\hat{B}),\hat{A})$ for all $\hat{A},\hat{B}\in \mathfrak{L}(\mathcal{H})$.
Therefore, positivity of $\tilde{\Lambda}$ implies the positivity of the adjoint map and vice versa.
This, in turn, implies that $0 \leq \text{Tr}\hat{\sigma}\hat{I}_i\otimes\Lambda_j(\hat{W})$: 
any positive map preserves the witnessing property of $\hat{W}$ with respect to the lower bounded inequality.

What can we say about the upper bound $\tilde{g}$?
Given that $\hat{I}_i\otimes\tilde{\Lambda}_j$ ($i,j=1,2$ and $i\neq j$) is a positive map, we have
\begin{equation}
\frac{1}{N}\text{Tr}\hat{\tau}\hat{W}\leq \tilde{g},
\end{equation}
where $N=\text{Tr}\hat{\tau}$ and $\frac{1}{N}\hat{\tau}$ is a legitimate separable quantum state.
Consequently, $\text{Tr}\hat{\tau}\hat{W}\leq \tilde{g}$ if and only if $\text{Tr}\hat{I}_i\otimes\tilde{\Lambda}_j(\hat{\sigma}) \leq 1$:
we are allowed to apply local maps which are trace-non-increasing positive (TnIP) to the set of separable states.
Passing to the positive adjoint map $\Lambda_j$ we have $\text{Tr}\hat{\sigma}\hat{I}_i\otimes\Lambda_j(\hat{W})\leq \tilde{g}$.
In addition, $\tilde{\Lambda}_j$ must preserve the Hermiticity which imposes the Kraus representation $\tilde{\Lambda}_j(\hat{A})=\sum_k \epsilon_{jk} \hat{E}_{jk} \hat{A}\hat{E}^\dag_{jk}$ with $\epsilon_{jk}=\pm 1$ so that the adjoint map will be given by $\Lambda_j(\hat{B})=\sum_k \epsilon_{jk} \hat{E}^\dag_{jk} \hat{B}\hat{E}_{jk}$~\cite{Min}.
Thus, the trace-non-increasing property gives
\begin{equation}
\begin{split}
&\text{Tr}\tilde{\Lambda}_j(\hat{A}) = \text{Tr}\sum_k \epsilon_{jk} \hat{E}^\dag_{jk}\hat{E}_{jk} \hat{A} \leq 1 \text{~for all~}\hat{A}\\
& \Leftrightarrow \sum_k \epsilon_{jk} \hat{E}^\dag_{jk}\hat{E}_{jk} \leq \hat{\mathbb{I}}
\Leftrightarrow \text{Tr} \Lambda_j(\hat{\mathbb{I}}) \leq \hat{\mathbb{I}},
\end{split}
\end{equation}
that is the adjoint map must be sub-unital.
In summary, we are allowed to apply local maps which are sub-unital positive (SUP) to any witness of the form~\eqref{Test} preserving the upper bound in Eq.~\eqref{Wineq}.

It is also worth noticing that, for both cases discussed above, the map $\Lambda_j$ should necessarily be entanglement-non-breaking (EnB) to preserve the witnessing capability of $\hat{W}$.
In particular, the partial transposition operation is a trace-non-increasing positive but not completely positive (and thus SUP) map which is also EnB, preserving the witnessing property of the operators in Eq~\eqref{Test}.

In the following, we show another important property of our construction; it is basis set independent and it is uniquely determined by the pseudo-metric $\mathsf{G}$.
Consider an Invertible map $\Lambda:\mathfrak{A}(\mathcal{H})\rightarrow\mathfrak{A}(\mathcal{H})$ and its dual $\Lambda':\mathfrak{A}(\mathcal{H})\rightarrow\mathfrak{A}(\mathcal{H})$ which is defined so that it preserves the duality condition~\eqref{orthXdX}:
\begin{equation}
\text{Tr}\Lambda^{\prime\dag}(X_{ij}^{\prime})\Lambda(X_{kl})=\delta_{ik}\delta_{jl}~\text{and}~\Lambda^{\prime\dag}(X_{ij}^{\prime})=\Lambda^{\prime}(X_{ij}^{\prime\dag}).
\end{equation}
Hence, it can be easily verified that 
\begin{equation}
(\hat{\sigma},\hat{W})=\text{Tr}\hat{\sigma}\hat{W}=\text{Tr}\hat{\sigma}\Lambda\otimes\Lambda^{\prime}(\hat{W})=\text{Tr}\hat{\varrho}_1\mathsf{G}(\hat{\varrho}_2).
\end{equation}
Thus, the witness operator in Eq.~\eqref{Test} is $\Lambda\otimes\Lambda^{\prime}$ invariant for any invertible map $\Lambda$.
It is also clear that all basis sets for $\mathfrak{L}(\mathcal{H})$ can be transformed to each other via an invertible map $\Lambda$---they are all isomorphic to each other.
In other words, by choosing a fixed metric $\mathsf{G}$, the outcome of the witnessing procedure is independent of the chosen bases.
The above considerations imply that the Hermitian operator $\hat{W}$ in Eq.~\eqref{Test} can be a bipartite entanglement witness for any choice of the basis set $\mathfrak{X}$.
Accordingly, we may consider a typical local orthogonal rank-one (self-dual) basis set $\mathfrak{X}=\{|i\rangle\langle j|\}_{i,j=1}^n$ and construct the EW $\hat{W}=\sum_{ij=1}^n|i\rangle\langle j|\otimes|i\rangle\langle j|$ where the metric is chosen to be the identity superoperator.
Any other basis set $\mathfrak{X}$ and its dual set $\mathfrak{X}'$ can be obtained from $\{|i\rangle\langle j|\}_{i,j=1}^n$ using two dual invertible maps, $\Lambda$ and $\Lambda^\prime$. 
Moreover, we may note that $\hat{W}$ is the Choi matrix of the identity map.
Notably, all arguments of this section equally hold for infinite dimensional Hilbert spaces,
because both the Riesz theorem and Cauchy-Schwarz inequality hold for any Banach space.

A relevant subtle point here is that, in the infinite or very high dimensional cases, one is not able to measure an infinite number of basis elements in practice.
We close this section by proving that our construction method can be equally applied to a subset of bases elements.

\begin{lemma}
Any operator of the form~\eqref{Test}, constructed from any subset of the basis set, $\{\hat{Y}_{ij}\}\subseteq\{\hat{X}_{ij}\}$, satisfies the necessary condition of the entanglement witness Eq.~\eqref{Wineq} provided that it spans a subset of the state space $\mathfrak{M}(\mathcal{H})\subseteq\mathfrak{S}(\mathcal{H})$.
\end{lemma}

\begin{proof}
Without loss of generality, we assume a finite dimensional Hilbert space.
The set $\mathfrak{M}(\mathcal{H})$ is convex and endowed with a set of rank-one projection generators $\{\hat{P}^{\rm m}_i\}$ in a one to-to-one correspondence with $\{\hat{Y}_{ij}\}$, $\{\hat{P}^{\rm m}_i\}\cong\{\hat{Y}_{ij}\}$.
Let us call $\mathfrak{M}^{\rm c}(\mathcal{H})$ the state space generated by the span of the basis set $\{\hat{Z}_{ij}\}=\{\hat{X}_{ij}\}\backslash\{\hat{Y}_{ij}\}$.
In the same way, $\mathfrak{M}^{\rm c}(\mathcal{H})$ is the convex hull of rank-one projections $\{\hat{P}^{\rm c}_i\}\cong\{\hat{Z}_{ij}\}$.
The whole convex set $\mathfrak{S}(\mathcal{H})$ is generated by $\{\hat{X}_{ij}\}=\{\hat{Y}_{ij}\}\bigcup\{\hat{Z}_{ij}\}\cong\{\hat{P}^{\rm m}_i\}\bigcup\{\hat{P}^{\rm c}_i\}$.
It is also clear that $\mathfrak{M}^{\rm c}(\mathcal{H})\perp\mathfrak{M}(\mathcal{H})$.
Therefore, there exists a decomposition $\hat{\varrho}=\alpha\hat{\varrho}^{\rm m}+(1-\alpha)\hat{\varrho}^{\rm c}$ with $\alpha\in[0,1]$, $\hat{\varrho}^{\rm m}\in\mathfrak{M}(\mathcal{H})$ and $\hat{\varrho}^{\rm c}\in\mathfrak{M}^{\rm c}(\mathcal{H})$ for all $\hat{\varrho} \in \mathfrak{S}(\mathcal{H})$.
We conclude the proof with noticing that the projection of any state $\hat{\varrho} \in \mathfrak{S}(\mathcal{H})$ onto $\mathfrak{M}(\mathcal{H})$ is the legitimate density operator $\hat{\varrho}^{\rm m}$, and the proof for the conditions of witness operator~\eqref{Test} holds for it.
That is, given the witness $\hat{W}^{\rm m}=\sum_{ij}g_{ij}\hat{Y}_{ij}\otimes\hat{Y}^{\prime\dag}_{ij}$ and for any product state $\hat{\sigma}=\hat{\sigma}_1\otimes\hat{\sigma}_2$ where $\hat{\sigma}_i=\alpha_i\hat{\sigma}^{\rm m}_i+(1-\alpha_i)\hat{\sigma}^{\rm c}_i$ with $\alpha_i\in[0,1]$, $\hat{\sigma}^{\rm m}_i\in\mathfrak{M}(\mathcal{H})$ and $\hat{\sigma}^{\rm c}_i\in\mathfrak{M}^{\rm c}(\mathcal{H})$ for $i=1,2$, one has
\begin{equation}
0\leq(\hat{\sigma},\hat{W}^{\rm m})={\rm Tr}\hat{\sigma}\hat{W}^{\rm m}=\alpha_1\alpha_2{\rm Tr}\hat{\sigma}^{\rm m}_1\hat{\sigma}^{\rm m}_2 \leq \tilde{g},
\end{equation}
The generalization to $\hat{\sigma}=\sum_{k}p_k\hat{\sigma}_{1;k}\otimes\hat{\sigma}_{2;k}$ with $\sum_{k}p_k{=}1$ is straightforward.
\end{proof}

%

\section{The sufficient criterion}~\label{SecIII}

In what follows, as the main result of this paper, we construct an EW and apply it to the output of a BS to obtain a sufficient criterion for entanglement generation and detection.

A useful witness operator of type~\eqref{Test} can be constructed from $s$-ordered displacement operator $\hat{D}(\alpha,s)=\hat{D}(\alpha)\exp(s|\alpha|^2/2)$ with $\alpha\in\mathbb{C}$ and $s\in[-1,1]$. 
The parameter values $s=-1,0,1$ correspond to anti-normal, Weyl-Wigner and normal orderings, respectively~\cite{Glauber}.
The $s$-parametrized characteristic function of the state $\hat{\varrho}$ is given by $\chi(\alpha,s)=\text{Tr}\hat{\varrho}\hat{D}(\alpha,s)$.
We also have the duality relation 
\begin{equation}
\text{Tr}\hat{D}^{\dag}(\alpha,-s)\hat{D}(\beta,s)=\pi\delta^{(2)}(\beta-\alpha),
\end{equation}
and thus the dual operator to $\hat{D}(\alpha,s)$ is $\hat{D}^{\dag}(\alpha,-s)=\hat{D}(-\alpha,-s)$.
Therefore, we may replace the sum in Eq.~\eqref{Test} by integration over the whole complex plane and construct the following $s$-parameterized witness operator,
\begin{equation}
\label{WDisp}
 \hat{W}(s)=\int \frac{d^2\alpha}{\pi} \hat{D}_1(\alpha,s)\otimes\hat{D}_2(-\alpha,-s).
\end{equation}
Here, the corresponding pseudo-metric $\mathsf{G}$ is chosen to be the identity, and $\hat{D}_1(\alpha,s)$ and $\hat{D}_2(\alpha,s)$ act on $\mathcal{H}_1$ and $\mathcal{H}_2$, respectively.

To find a sufficient criterion for entanglement generation by a BS, we apply a specially modified version of the witness~\eqref{WDisp} to the output of a BS.
Then, using a retrodictive calculation we investigate the violation of the witness inequality in terms of the input sates; cf. Fig.~\ref{MZint} (a).
\begin{figure*}
\includegraphics[width=12cm]{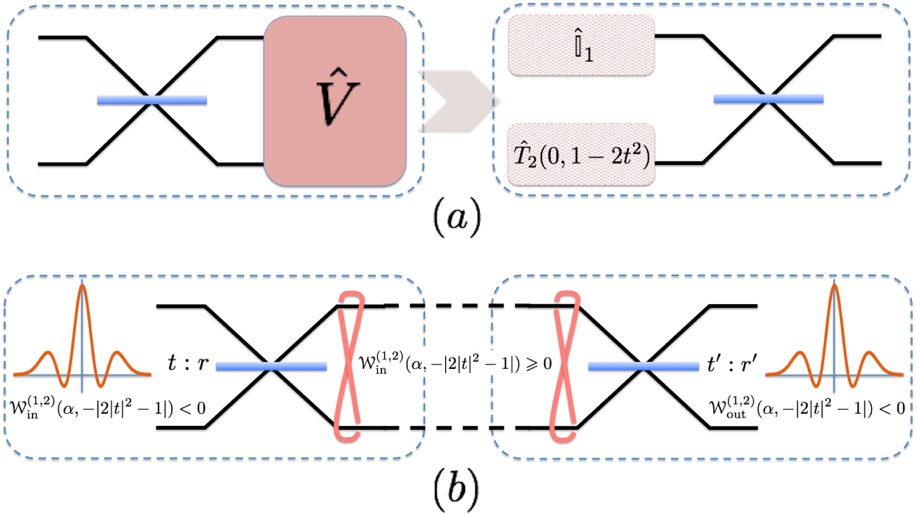}
  \caption{(a) Application of the EW of Eq.~\eqref{VDisp} to the output of a BS.
  The witness operator $\hat{V}$ is designed in such a way that it results in the identity operator for mode one after a retrodictive calculation of the BS operation.
  (b) The local detection of entanglement by detection of marginal nonclassicality.
  Note that, after the first $(t:r)$-BS the marginal QPDs of the intermediate state are positive for $s\leq-|2|t|^2-1|$.
  The second $(t':r')$-BS is chosen such that the maximum negativity is retrieved for one of the marginal output QPDs.
  }\label{MZint}
\end{figure*} 
First, we consider the witness operator~\eqref{WDisp} for $s=1$.
Next, we apply the positive map
\begin{equation}\label{LambdaMap}
(\mathsf{T} \circ \Lambda)(\hat{D}(\alpha,-1)):= \frac{t^2}{r^2}\hat{D}(-\frac{t}{r}\alpha,-1)~\text{with}~\frac{t}{r}\geq 1
\end{equation}
to the second mode of $\hat{W}(1)$ where $\mathsf{T}$ is the transposition map acting as $\alpha\rightarrow -\alpha$ and, without loss of generality, we assume that $t,r\in\mathbb{R}$ and $t^2+r^2=1$; c.f. Appendix~\ref{TnIPApp} for the proof of positivity of the map~\eqref{LambdaMap}.
Notice that, as we will be testing the lower bounded witnessing inequality, positivity of the map $\Lambda$ is sufficient for the upcoming analysis.
This gives the modified witness
\begin{equation}
\label{VDisp}
 \hat{V}=\frac{t^2}{r^2}\int \frac{d^2\alpha}{\pi} \hat{D}_1(\alpha,1)\otimes\hat{D}_2(\frac{t}{r}\alpha,-1).
\end{equation}
Finally, given a BS of transitivity $t$ and reflectivity $r$ -- in short a ($t:r$)-BS, we apply the witness $\hat{V}$ to its output.
One finds the expectation value of the witness operator for any input state to be
\begin{align}\label{rhoOut2}
\begin{split}
	\text{Tr}\hat{\varrho}_{\rm out}\hat{V}&=\text{Tr}\hat{U}_{t:r}\hat{\varrho}_{\rm in}\hat{U}^{\dag}_{t:r}\hat{V}\\
	&=t^2 \text{Tr}\hat{\varrho}_{\rm in}\hat{\mathbb{I}}_1\otimes\hat{T}_2(0,1-2t^2)\\
	&=\pi t^2 \mathscr{W}_{\rm in}^{(2)}(0,1-2t^2),
\end{split}
\end{align}
where we have used the beam splitter transformation $\hat{U}_{t:r}^{\dag}(\hat{a},~\hat{b})\hat{U}_{t:r}=(t\hat{a}+r\hat{b},~-r\hat{a}+t\hat{b})$ together with the definition
\begin{equation*}
\hat{T}(\beta,s)=\int \frac{d^2 \alpha}{\pi} e^{\beta\alpha^\ast - \beta^\ast\alpha} \hat{D}(\alpha,s)
\end{equation*}
for the $s$-parametrized Wigner operator such that ${\rm Tr}\hat{\varrho}\hat{T}(\beta,s)=\pi\mathscr{W}_{\hat{\varrho}}(\beta,s)$ is the $s$-parametrized quasiprobability distribution (QPD) of the state $\hat{\varrho}$~\cite{Glauber}.
As a result, $\mathscr{W}_{\rm in}^{(2)}(0,1-2t^2)<0$ if and only if $\text{Tr}\hat{\varrho}_{\rm out}\hat{V}<0$.
Moreover, if $\mathscr{W}_{\rm in}^{(2)}(\beta,1-2t^2)<0$ at some phase-space point $\beta$, we may shift the witness operator using an appropriate local displacement operation to the point $\beta$, $\hat{V}(\beta)$, and obtain a negative value for $\text{Tr}\hat{\varrho}_{\rm out}\hat{V}(\beta)$ without changing the entanglement content of the state.
To obtain the similar condition on the marginal QPDs of the first input, we only need to apply the map $\Lambda$ of Eq.~\eqref{LambdaMap} (without the transposition $\mathsf{T}$) to the first mode of $\hat{W}(-1)$ to get the modified witness $\hat{V}'$.
A retrodictive evaluation at the input thus gives
\begin{align*}
	\text{Tr}\hat{\varrho}_{\rm out}\hat{V}'=\pi t^2 \mathscr{W}_{\rm in}^{(1)}(0,1-2t^2),
\end{align*}
implying that $\mathscr{W}_{\rm in}^{(1)}(0,1-2t^2)<0$ if and only if $\text{Tr}\hat{\varrho}_{\rm out}\hat{V}'<0$ which can be further generalized to any phase-space point $\beta$.

For the case of $t/r<1$, although the map $\Lambda$ of Eq.~\eqref{LambdaMap} is no longer positive, we can alternatively apply the legitimate map
\begin{equation}\label{Lambda2Map}
\Xi(\hat{D}(\alpha,-1)):= \frac{r^2}{t^2}\hat{D}(\frac{r}{t}\alpha,-1)~\text{with}~\frac{r}{t}> 1
\end{equation}
to the second mode of $\hat{W}(1)$ to get a third witness $\hat{V}''$.
This gives the counterpart of Eq.~\eqref{rhoOut2} as
\begin{align}\label{rhoOut1}
	\text{Tr}\hat{\varrho}_{\rm out}\hat{V}''=\pi r^2 \mathscr{W}_{\rm in}^{(1)}(0,1-2r^2),
\end{align}
and thus, negativity of $(1-2r^2)$-parametrized marginal of the first input implies the entanglement of the output.
To obtain a similar condition on the second input mode in this case, one should apply $(\mathsf{T}\circ\Xi)(\hat{D}(\alpha,-1))$ to the first mode of $\hat{W}(-1)$ to get the appropriate witness.

Taking into account the possibility of a complex valued transitivity, we can combine all the conditions above into one negativity condition: $\mathscr{W}^{(1,2)}_{\rm in}(\alpha,-|2|t|^2-1|)<0$.
In short, negativities of at least one of the marginal $s$-parametrized QPDs for $s=-|2|t|^2-1|$ at the input (output) to a $(t:r)$-BS implies the entanglement of the output (input).

\section{Discussion of the results}\label{SecIV}

There are several interesting implications by the results in previous section which we will discuss one by one.
The first important thing to notice is that, our condition is a sufficient one.
That is, for instance, a negative marginal Wigner function at the input to a balanced BS guarantees the entanglement of the output modes, nevertheless, it is not necessary for inputs to possess nonclassical Wigner functions to produce entangled outputs.
The second important message is that, concerning the properties of the input state, one can see that there is no need to inject pure, product, or separable states into the BS.
The input can even be entangled, and if one of the marginal QPDs satisfies the above condition the entanglement will be preserved by the BS.
More interestingly, it is sufficient that the criterion is satisfied by only one of the input modes, regardless of what quantum state is being injected into the other port of the BS.
The conclusion from this fact is that a pure single photon will generates entanglement upon mixing with any quantum state on any BS, because except its $Q$-function, corresponding to a fully transitive or reflective BS, all of its QPDs possess negativities.
In particular, in an optical arrangement, high photon number entangled states are of special interest, because their large number of photons make them efficient for interactions.
A way for obtaining such states is by mixing a single photon with a coherent state on a BS~\cite{Sekatski}.
Surprisingly, using our criterion, it turns out that the input state to the other port of the BS can be any state, e.g. a bright thermal state of arbitrarily high temperature.
In this scenario, the high number of photons from the thermal state split into two classically distinguishable rays while getting entangled due to the negativity of the injected single photon.
This is just like nonclassicality being imprinted onto the output modes as entanglement. 

Let us also discuss the inverse case, where entanglement of the input state imposes the marginal nonclassicality on the outputs.
The main problem in applying EWs to CV systems is the extreme difficulty of coincidence measurements due to the high number of dimensions.
Nevertheless, there are entangled states for which one does not need coincidence measurements to verify their entanglement when the verification is done locally.
It suffices to interact the two rays on an arbitrary BS, perform homodyne measurements separately on both outputs and check the negativity of the reduced QPDs which is much easier; cf. Fig.~\ref{MZint} (b).
That is, we may reverse the generation process and verify the entanglement of the input.
In practice, there is even no need for a full state reconstruction to verify nonclassicality of QPDs, as they can also be detected using appropriate witnesses much more simply than entanglement~\cite{nonCW}.
Operationally speaking, this technique contains no measurement of correlated events between the two modes.
Another lesson we learn from this fact is that no matter how much nonclassicality has been separably injected into the BS, the output will always exhibit classical $s$-parametrized marginal QPDs for $s\leq-|2|t|^2-1|$.
In this way, negativities of the QPDs can be considered as information carriers about the entanglement of the source.

\section{The evolution of the entanglement}~\label{SecV}

As a final word, we consider the evolution and verification of the entanglement generated via satisfying our nonclassicality criterion when the state traverses through a lossy Mach-Zehnder interferometer (MZI) (see e.g. Fig.~\ref{MZint}~(b)).
Suppose that we inject a state with nonclassical Wigner function into the first port of a balanced BS to generate entanglement and expect to receive it at the first port of the output ($t:r$)-BS.
It is known that a lossy channel will transform anti-normally ordered displacement operators as $\Lambda_{\eta}(\hat{D}(\alpha,-1))=\hat{D}(\alpha/\eta,-1)/\eta^2$ where $0\leq\eta\leq 1$ quantifies the loss~\footnote{$\eta=0$ represents a completely lossy channel while $\eta=1$ is used for a perfectly lossless channel.}.
At the other end of the MZI, we are interested in parameters for the output BS such that most of the effect from the state in the second input port is canceled.
We have shown in details in the Appendix~\ref{App} that the best choice of the output BS parameters, $t$ and $r$, are
\begin{equation}\label{tANDr}
 t=\frac{\eta_2}{\sqrt{\eta_1^2+\eta_2^2}}\quad\text{and}\quad r=\frac{\eta_1}{\sqrt{\eta_1^2+\eta_2^2}},
\end{equation}
where $\eta_{1}$ and $\eta_{2}$ are the losses of the first and second arms, respectively.
Consequently, the detected marginal output is
\begin{equation}\label{GDeg}
\begin{split}
\hat{\varrho}_{1;{\rm out}}={\rm Tr}_{2}\hat{\varrho}_{\rm out} =&\int \frac{d^2\alpha}{\pi} \chi^{(1)}_{\rm in}(-\frac{\sqrt{2}\eta_1\eta_2}{\sqrt{\eta^2_1+\eta^2_2}}\alpha)\\
&\times \exp\{-\frac{(\eta_1-\eta_2)^2}{2(\eta_1^2+\eta_2^2)}|\alpha|^2\} \hat{D}_1(\alpha),
\end{split}
\end{equation}
where $\chi^{(1)}_{\rm in}(\alpha)={\rm Tr}\hat{\varrho}_{\rm in}\hat{D}_1(\alpha)$.
Equation~\eqref{GDeg} clearly represents a Gaussian degrading of the output to a QPD of $s=-(\eta_1-\eta_2)^2/(\eta_1^2+\eta_2^2)$.
The very interesting situation would then occur if the loss is symmetric in both arms ($\eta_1=\eta_2=\eta$)
for which using Eq.~\eqref{tANDr} the output BS should be balanced, there will be no Gaussian smoothing effect, and we get
\begin{equation}
\begin{split}
	\mathscr{W}^{(1)}_{\rm out}(\alpha)=\frac{1}{\eta^2}\mathscr{W}^{(1)}_{\rm in}(\frac{\alpha}{\eta}).
\end{split}
\end{equation}
This represents the regime in which the negativity degrades quadratically with loss as $\eta$ goes to zero (total loss in both arms).
In fact, as long as there is not $100\%$ loss in both channels, any negativity at the input will ultimately survive and can be detected at the output, although it can be very small.
The effect shown here is a generalization of the similar phenomena for single photon and Gaussian entanglement where very large losses in each mode makes the entanglement vanishingly small, but do not destroy it until the mode is completely blocked.
The example of a lossy single-photon entangled state is given in Appendix~\ref{SPApp}.

\section{Summary and conclusions}~\label{SecVI}

In conclusion, we found a sufficient nonclassicality condition for an arbitrary beam splitter (BS) to generate entanglement.
We also found a sufficient nonclassicality condition for the marginal outputs of a BS to ensure entanglement of the input state.
To achieve this goal, first, we introduced a proper entanglement witness (EW) construction method from arbitrary local bases sets of operators.
We have extensively discussed the properties of our construction scheme under bases transformations and used them to improve EWs.
Next, we defined a continuous variable EW capable of detecting the negativities of the marginal quasiproability distributions input to an arbitrary BS.
We have proven that negativity of at least one of the $s$-parametrized marginal quasiprobabilities for $s=-|2|t|^2-1|$ at the input (output) of a BS with transitivity $t$ implies the entanglement of the output (input) state.
In particular, we showed that a single photon of sufficient purity can transform \emph{any} input state into a bipartite entangled state upon interaction on an arbitrary BS.
We also showed that this entanglement can be imprinted on the marginal output quasiproabilities of a second BS in the form of negativities to certify the entanglement of the input modes.
In other words, our sufficient criterion provides an easy way for the generation and detection of entangled states with arbitrary high number of photons.
Using general arguments, we studied the evolution of the entanglement generated and detected by our criterion in a lossy Mach-Zehnder interferometer and extended the important fact that the entanglement does not diminish as long as both modes are not completely blocked (lost).
Last but not least, our results hold true for other bosonic systems such as quantum opto-mechanics and spin ensembles as well.

\section*{Acknowledgements}
The authors gratefully acknowledge valuable discussions by S. Rahimi-Keshari, J. Sperling, M. R. Vanner, and T. C. Ralph.
FS was supported through the Australian Research Council Discovery Project (No. DP140101638) and the Australian Research Council Centre of Excellence of Quantum Computation and Communication Technology (Project No. CE110001027)

\appendix
\begin{widetext}

\section{The positive map $\bf\Lambda$}\label{TnIPApp}
To show that the map in Eq.~\eqref{LambdaMap} is positive, let us consider the general case of a map
\begin{equation}
\Lambda(\hat{D}(\alpha,-1)):=\kappa^{2} \hat{D}(\kappa\alpha,-1).
\end{equation}
It is known that~\cite{Glauber}
\begin{equation}\label{FtransD}
\hat{T}(\beta,-1)=\int \frac{d^2 \alpha}{\pi} e^{\beta\alpha^\ast - \beta^\ast\alpha} \hat{D}(\alpha,-1) = |\beta\rangle\langle\beta|,
\end{equation}
where $|\beta\rangle$ is the coherent state.
Applying the map $\Lambda$ to Eq.~\eqref{FtransD}, we get
\begin{equation}
\begin{split}
\Lambda(|\beta\rangle\langle\beta|) = \Lambda(\hat{T}(\beta,-1)) & = \kappa^2 \int \frac{d^2 \alpha}{\pi} e^{\beta\alpha^\ast - \beta^\ast\alpha} \hat{D}(\kappa\alpha,-1)\\
& = \int \frac{d^2 \gamma}{\pi} e^{(\frac{\beta}{\kappa})\gamma^\ast - (\frac{\beta}{\kappa})^\ast\gamma} \hat{D}(\gamma,-1)\\
&=|\frac{\beta}{\kappa}\rangle\langle\frac{\beta}{\kappa}|.
\end{split}
\end{equation} 
Clearly, the map $\Lambda$ acts as the attenuation channel and it is physically legitimate and thus positive if and only if $\kappa \geq 1$.
It is also important to note that the (partial) transposition operation $\mathsf{T}$ in Eq.~\eqref{LambdaMap} is also positive, and that the composition of two positive maps is also positive.

\section{State evolution in a lossy MZI}\label{App}
Suppose that we are injecting two separable states into our balanced beam splitter (BS) to generate entanglement.
Thus, from our theorem, we need a negative Wigner function at least in one of the input ports.
Assume that the negative input is injected into port one. 
Then, the state exiting the input BS will be
\begin{align}
	\hat{\varrho}_{\rm int}=\int \frac{d^2\alpha d^2\beta}{\pi^2} \chi^{(1)}_{\rm in}(-\alpha,1)\chi^{(2)}_{\rm in}(-\beta,1) \hat{D}_1(\frac{\alpha+\beta}{\sqrt{2}},-1)\otimes\hat{D}_2(\frac{-\alpha+\beta}{\sqrt{2}},-1),
\end{align}
in which $\hat{D}_i(\alpha,s)=\hat{D}_i(\alpha)e^{\frac{s}{2}|\alpha|^2}$ is the $s$-ordered displacement operator, and $\chi^{(i)}(\alpha,s)={\rm Tr}\hat{\varrho}\hat{D}_i(\alpha,s)$ is the $s$-parameterized characteristic function of the $i$th mode for $i=1,2$.
Using the effect of a lossy channel on the anti-normally ordered displacement operators as described in Appendix~\ref{TnIPApp},
\begin{equation}
	\Lambda_{\eta}(\hat{D}(\alpha,-1))=\frac{1}{\eta^2}\hat{D}(\frac{\alpha}{\eta},-1), \quad 0\leq\eta\leq 1,
\end{equation}
after losses in both channels we get
\begin{align}
	\hat{\varrho}_{\rm int}=\frac{1}{\eta_1^2\eta_2^2}\int \frac{d^2\alpha d^2\beta}{\pi^2} \chi^{(1)}_{\rm in}(-\alpha,1)\chi^{(2)}_{\rm in}(-\beta,1) \hat{D}_1(\frac{\alpha+\beta}{\sqrt{2}\eta_1},-1)\otimes\hat{D}_2(\frac{-\alpha+\beta}{\sqrt{2}\eta_2},-1).
\end{align} 

Now, we want to retrieve the negativity on the detection site.
We put the beams on a BS of transitivity $t$.
The state after the output BS is
\begin{equation}
\begin{split}
	\hat{\varrho}_{\rm out}=\frac{1}{\eta_1^2\eta_2^2}\int \frac{d^2\alpha d^2\beta}{\pi^2} \chi^{(1)}_{\rm in}(-\alpha,1)\chi^{(2)}_{\rm in}(-\beta,1) & \hat{D}_1\Big([\frac{t}{\eta_1}+\frac{r}{\eta_2}]\frac{\alpha}{\sqrt{2}}+[\frac{t}{\eta_1}-\frac{r}{\eta_2}]\frac{\beta}{\sqrt{2}},-1\Big)\\
	&\otimes \hat{D}_2\Big([\frac{r}{\eta_1}-\frac{t}{\eta_2}]\frac{\alpha}{\sqrt{2}}+[\frac{r}{\eta_1}+\frac{t}{\eta_2}]\frac{\beta}{\sqrt{2}},-1\Big).
\end{split}
\end{equation}
We are only interested in the reduced quaiprobabilities of the output modes.
Also, by our choice of phases of the beam splitters, we know that the negativity should appear in the output mode one.
So, we trace out the second mode:
\begin{align}\label{Calc}
	\hat{\varrho}_{1;\rm out}&={\rm Tr}_{2}\hat{\varrho}_{\rm out}\nonumber\\
	&=\frac{1}{\eta_1^2\eta_2^2}\int \frac{d^2\alpha d^2\beta}{\pi} \chi^{(1)}_{\rm in}(-\alpha,1)\chi^{(2)}_{\rm in}(-\beta,1) \hat{D}_1\Big([\frac{t}{\eta_1}+\frac{r}{\eta_2}]\frac{\alpha}{\sqrt{2}}+[\frac{t}{\eta_1}-\frac{r}{\eta_2}]\frac{\beta}{\sqrt{2}},-1\Big)\delta^{(2)}\Big([\frac{r}{\eta_1}+\frac{t}{\eta_2}]\frac{\beta}{\sqrt{2}}-[\frac{t}{\eta_2}-\frac{r}{\eta_1}]\frac{\alpha}{\sqrt{2}}\Big)\nonumber\\
	&=\frac{2}{(t\eta_1+r\eta_2)^2}\int \frac{d^2\alpha d^2\gamma}{\pi} \chi^{(1)}_{\rm in}(-\alpha,1)\chi^{(2)}_{\rm in}(-\frac{\sqrt{2}\gamma}{[\frac{r}{\eta_1}+\frac{t}{\eta_2}]},1) \hat{D}_1\Big([\frac{t}{\eta_1}+\frac{r}{\eta_2}]\frac{\alpha}{\sqrt{2}}+\frac{t\eta_2-r\eta_1}{t\eta_1+r\eta_2}\gamma,-1\Big)\delta^{(2)}\Big(\gamma-[\frac{t}{\eta_2}-\frac{r}{\eta_1}]\frac{\alpha}{\sqrt{2}}\Big)\nonumber\\
	&=\frac{2}{(t\eta_1+r\eta_2)^2}\int \frac{d^2\alpha}{\pi} \chi^{(1)}_{\rm in}(-\alpha,1)\chi^{(2)}_{\rm in}(-\frac{t\eta_1-r\eta_2}{t\eta_1+r\eta_2}\alpha,1) \hat{D}_1\Big(\frac{\sqrt{2}}{t\eta_1+r\eta_2}\alpha,-1\Big).
\end{align}
In the second line of Eq.~\eqref{Calc}, we have used $\gamma:=[\frac{r}{\eta_1}+\frac{t}{\eta_2}]\frac{\beta}{\sqrt{2}}$. We want to cancel the effect of the second input mode, so we choose
\begin{align}
	&t\eta_1-r\eta_2=0\Rightarrow t=\frac{\eta_2}{\sqrt{\eta_1^2+\eta_2^2}}\quad\text{and}\quad r=\frac{\eta_1}{\sqrt{\eta_1^2+\eta_2^2}},\quad (t^2+r^2=1),\\
	&t\eta_1+r\eta_2=\frac{2\eta_1\eta_2}{\sqrt{\eta_1^2+\eta_2^2}},
\end{align}
and thus,
\begin{align}\label{1Out}
	\hat{\varrho}_{1;\rm out}&=\frac{\eta_1^2+\eta_2^2}{2\eta_1^2\eta_2^2}\int \frac{d^2\alpha}{\pi} \chi^{(1)}_{\rm in}(-\alpha,1) \hat{D}_1\Big(\frac{\sqrt{\eta_1^2+\eta_2^2}}{\sqrt{2}\eta_1\eta_2}\alpha,-1\Big).
\end{align}
Multiplying both sides of Eq.~\eqref{1Out} by $\hat{D}_1(-\beta,1)$ gives
\begin{align}
	&\chi^{(1)}_{\rm out}(-\beta,1)=\chi^{(1)}_{\rm in}(-\frac{\sqrt{2}\eta_1\eta_2}{\sqrt{\eta^2_1+\eta^2_2}}\beta,1),\\
	\Rightarrow &\chi^{(1)}_{\rm out}(\beta)=\chi^{(1)}_{\rm in}(\frac{\sqrt{2}\eta_1\eta_2}{\sqrt{\eta^2_1+\eta^2_2}}\beta)\exp\{-\frac{(\eta_1-\eta_2)^2}{2(\eta_1^2+\eta_2^2)}|\beta|^2\}.
\end{align}
This represents the fact that, in the first output, we get a Wigner function which is smoothed by the parameter
\begin{equation}
	s=-\frac{(\eta_1-\eta_2)^2}{\eta_1^2+\eta_2^2}.
\end{equation}
If the loss is symmetric in both arms ($\eta_1=\eta_2=\eta$), there will be no smoothing effect and we get
\begin{equation}
	\chi^{(1)}_{\rm out}(\beta)=\chi^{(1)}_{\rm in}(\eta\beta).
\end{equation}
Equivalently, in Fourier space
\begin{equation}
	\mathscr{W}^{(1)}_{\rm det}(\beta)=\frac{1}{\eta^2}\mathscr{W}^{(1)}_{\rm in}(\frac{\beta}{\eta}).
\end{equation}

\section{The effect of loss on the single-photon entangled state}\label{SPApp}

The result of Sec.~\ref{SecIV} is an extension of the following simple example: a single photon splitting on a balanced BS.
In this case, where the other input to the BS is just the vacuum state, the state after the losses in each channel can be calculated easily.
The output of the BS is simply the Bell state $|\phi^+\rangle=(|10\rangle_{1,2}+|01\rangle_{1,2})/\sqrt{2}$.
Each mode suffers from a loss of $\eta_1=\eta_2=\eta$ ($0\leq\eta\leq 1$) which can be modeled by two BSs of transitivity $\eta$ and two ancillary modes.
The overall state will thus be given by
\begin{equation}
|\psi(\eta)\rangle=\frac{1}{\sqrt{2}}(|\phi^+(\eta)\rangle_{13}|00\rangle_{24} + |00\rangle_{13}|\phi^+(\eta)\rangle_{24}),
\end{equation}
where the modes $3$ and $4$ are the ancillae and $|\phi^+(\eta)\rangle_{ij}=\eta|10\rangle_{ij}+\sqrt{1-\eta^2}|01\rangle_{ij}$.
Now, tracing out the ancillae modes gives
\begin{equation}
\hat{\varrho}_{12}(\eta) = (1-\eta^2) |00\rangle\langle 00| + \frac{\eta^2}{2}|10\rangle\langle 10| +
\frac{\eta^2}{2}|01\rangle\langle 01| + \frac{\eta^2}{2}|01\rangle\langle 10| + \frac{\eta^2}{2}|10\rangle\langle 01|.
\end{equation}
Now, we can apply the partial transposition (PT) criterion in which the partially transposed state can be represented in matrix form as
\begin{equation}
\hat{\varrho}^{\Gamma}_{12}(\eta)=
\begin{pmatrix}
1-\eta^2 & 0 & 0 & \frac{\eta^2}{2} \\ 
 0 & \frac{\eta^2}{2} & 0 & 0 \\ 
 0 & 0 & \frac{\eta^2}{2} & 0 \\ 
\frac{\eta^2}{2}& 0 & 0 & 0
\end{pmatrix},
\end{equation}
where the partial transposition $\Gamma$ is taken with respect to mode two.
Note that for two-qubit states ($2\times 2$ quantum systems) PT criterion is both necessary and sufficient to verify the entanglement.
The eigenvalues of $\hat{\varrho}^{\Gamma}_{12}(\eta)$ are given by
\begin{equation}
\begin{split}
&\lambda_{1,2}=\frac{\eta^2}{2} \quad\text{(two-fold degenerate)},\\
&\lambda_3=\frac{1}{2}(1-\eta^2+\sqrt{2\eta^4-2\eta^2+1}),\\
&\lambda_4=\frac{1}{2}(1-\eta^2-\sqrt{2\eta^4-2\eta^2+1}).
\end{split}
\end{equation}
The eigenvalue $\lambda_4$ takes on negative values for any $\eta>0$, since $2\eta^4-2\eta^2+1\geq (1-\eta^2)^2$.
This means that for any value of loss below $100\%$, the lossy output state $\hat{\varrho}_{12}(\eta)$ is NPT-entangled and thus distillable.

\end{widetext}

\end{document}